\documentclass[11pt]{article}

\usepackage{times}
\usepackage{fullpage}
\usepackage{subfigure}
\usepackage{color}
\definecolor{gray}{rgb}{0.5,0.5,0.5}

\newtheorem{theorem}{Theorem}
\newtheorem{lemma}[theorem]{Lemma}

\usepackage{graphicx, epsfig, amsmath, latexsym, cite}
\newcommand{\qed}{\hfill \ensuremath{\Box}}
\newenvironment{proof}{{\bf Proof:}}{$\qed$\par}

\begin{document}

\title{Optimizing Tile Concentrations to Minimize Errors and Time for DNA Tile
Self-Assembly Systems}
\author{
Ho-Lin Chen
\thanks{
Center for Mathematics of Information, California Institute of Technology, Pasadena, CA 91101, USA. Email: holinc@gmail.com.}
\and Ming-Yang Kao
\thanks{Department of Electrical Engineering and Computer Science,
Northwestern University, Evanston, IL 60208, USA. Email: kao@northwestern.edu. This author's work was supported in part by NSF Grant CCF-1049899.}
}
\date{}

\begin{titlepage}
\maketitle
\thispagestyle{empty}
\begin{abstract} 
  \small\baselineskip=12pt
DNA tile self-assembly has emerged as a rich and promising primitive for nano-technology. This paper studies the problems of minimizing assembly time and error rate by changing the tile concentrations because changing the tile concentrations is easy to implement in actual lab experiments. We prove that setting the concentration of tile $T_i$ proportional to the square root of $N_i$ where $N_i$ is the number of times $T_i$ appears outside the seed structure in the final assembled shape minimizes the rate of growth errors for rectilinear tile systems. We also show that the same concentrations minimize the expected assembly time for a feasible class of tile systems. Moreover, for general tile systems, given tile concentrations, we can approximate the expected assembly time with high accuracy and probability by running only a polynomial number of simulations in the size of the target shape.

\end{abstract}
\end{titlepage}

\section{Introduction}

Considerable modern research in science and engineering has aimed to control smaller and smaller systems in many fields, including computer science and material science. As the size of a system approaches the molecular scale, precise direct external control becomes prohibitively costly, if not impossible. As a result, bottom-up self-assembly has emerged as a rich and promising primitive for nano-technology. In particular, DNA has received much attention as a substrate for molecular self-assembly because its combinatorial nature enables the programming of molecular behaviors by choosing appropriate DNA sequences to encode information. In addition, lab techniques for the manipulation of DNA are already well developed. For these considerations, DNA self-assembly has been proposed for a variety of applications, e.g., as a means to perform computation~\cite{brw05:counter,w98:phd,sszw06:circuits}, construct molecular patterns~\cite{rpe04:sierexp,r06:origami,zs94:octahedron,sqj04:octahedron,ddlhgs09:3D_origami,hdw:twisted_origami}, and build nano-scale machines~\cite{ytmsn00:tweezer,ss04:walker,sp04:walker,bk07:tweezer,gbt08_walker,ds06:3dcasette}.

DNA tiles which self-assemble according to simple rules have been developed in lab~\cite{wlws98:stripe} and mathematically analyzed based on the combinatorial tile assembly model (aTAM) proposed by Rothemund and Winfree~\cite{rw00:square1}. Under this model, there is a set of square tiles with a {\it glue} on each of the four edges. Each glue has a certain affinity for itself called {\it strength}. The self-assembly process starts from a distinguished {\it seed structure}. Assembly proceeds as tiles attach to the partially assembled structure (initially, just the seed structure) one by one when the combined strength of matched glues between a tile and the partial structure is at least the {\it temperature} of the tile system. Many interesting tile systems have been designed under aTAM, including systems that build counters~\cite{acgh01:square,cgm04:optcounter} and squares~\cite{rw00:square1,ks06:temperture,d09:randomized}, perform Turing-universal computation~\cite{w98:phd}, and produce arbitrary computable shapes~\cite{ll99:nphard,sw07:shapes}. Unfortunately, in laboratory settings, several events that aTAM does not model have been frequently observed. These events are referred to as {\it errors} in the tile self-assembly process. A more realistic stochastic model called the kinetic tile assembly model (kTAM) was proposed by Winfree~\cite{w98:phd} to describe the rates of these errors. The kTAM model calculates the rates for various types of attachments and detachments of tiles based on thermodynamics. 

In order to make DNA tile self-assembly practical, there are two important factors that need to be minimized, namely, the error rate and the time of the assembly process. One approach to reducing the error rate of a tile assembly system~\cite{wb02:proof,cg04:snaked,rsy04:compact,scw08:healing_proofreading,clg07:criteria} is to convert an existing error-prone tile system to a more robust tile system that assembles into the same shape or pattern up to scaling. These error correcting techniques increase the number of tile types by a multiplicative factor and thus are hard to implement in practice. In contrast, it is easy to change the concentrations of tiles. Therefore, it is natural to consider reducing the error rate by changing the concentrations of tiles. This approach has been studied using computer simulations and lab experiments. However, no closed-form formulas or efficient algorithms for finding the optimal tile concentrations have been previously found. It is also natural to consider changing the tile concentrations in order to minimize the assembly time. Adleman et al.~\cite{acghkmr02:opt} designed an algorithm to find tile concentrations that approximate the minimum expected assembly time within an $O(\log n)$ factor. Goel et al.~\cite{cgm04:optcounter} showed that if all tiles have equal concentrations, then the expected assembly time is proportional to the longest length of a path in the assembly order of the target shape. Also, some studies employed computer simulations to characterize the trade-offs between the time and the error rate of an assembly process~\cite{wb02:proof,cg04:snaked}. 

\paragraph{Our Results}
On the problem of minimizing the error rate, we formulate the rate of growth errors in terms of tile concentrations based on the kinetic tile assembly model. Using our formulation, we show that setting the concentration of each tile $T_i$ proportional to the square root of the number of times $T_i$ appears outside the seed structure in the target shape minimizes the rate of growth errors. This result holds for all rectilinear tile systems (i.e., tile systems that have the same growth directions for all tiles fixed throughout the assembly process) as well as many other systems that have been implemented in lab~\cite{sw05:zigzagexp,bsrw09:counter_origami}. We also have simulation results showing that facet errors can significantly affect the accuracy of the optimal tile concentrations predicted by our mathematical analysis. On the problem of minimizing the assembly time, we prove that the above concentrations for minimizing the rate of growth errors also minimizes the expected assembly time for a feasible class of tile systems. Moreover, for general tile systems, given tile concentrations, we show that the average assembly time over a polynomial number of simulations in the size of the target shape can approximate the expected assembly time with high accuracy and probability. 

The remainder of this paper is organized as follows. Section~\ref{sec:definition} describes the two tile assembly models that we use. Section~\ref{sec:error} contains the theoretical results on minimizing the rate of growth errors. Section~\ref{sec:simulation} contains the simulation results on growth errors and some discussion on facet errors. Section~\ref{sec:time} contains the theoretical results on estimating and minimizing the expected assembly time. Section~\ref{sec:conclusion} concludes the paper with some open problems.

\section{Two Tile Assembly Models}
\label{sec:definition}

\paragraph{The Combinatorial Tile Assembly Model}

The combinatorial tile assembly model was proposed by Rothemund and Winfree~\cite{rw00:square1}. It extends the theoretical model of tiling by Wang \cite{w61:tiles} to include a mechanism for growth based on the physics of molecular self-assembly. Informally, a tile self-assembly system has a set of tiles, each of which is a square with glues of various types on each of the four edges. Two tiles will stick to each other if they have compatible glues. Below we present a succinct definition of this model with minor modifications for ease of explanation.

A {\it tile} is an oriented unit square with the {\it north}, {\it east}, {\it south} and {\it west} edges labeled from some alphabet $\Sigma$ of {\it glues}. For each tile $t$, the glues of its four edges are denoted as $\sigma_N (t)$, $\sigma_E (t)$, $\sigma_S (t)$, and $\sigma_W (t)$. We describe a tile $t$ as the quadruple $(\sigma_N (t), \sigma_E (t), \sigma_S (t), \sigma_W (t))$. Consider the triple $<$$T, g , \tau$$>$ where $T$ is a finite set of tiles, $\tau \in {\bf Z}_{> 0}$ is the {\em temperature}, and $g$ is the {\it glue strength} function from $\Sigma \times \Sigma$ to ${\bf  Z_{\geq 0}}$. It is assumed that for all $x,y \in \Sigma$, the inequality $x\neq y$ implies $g(x,y) = 0$ and there is a glue {\em null} $\in \Sigma$, such that $g(x,null)=g(null,x)=0$ for all $x \in \Sigma$. A {\em configuration} is a map from ${\bf Z}^2$ to $T \bigcup empty$, where $empty$ is a special symbol indicating the absence of any tile.

A {\it tile system} is a quadruple ${\bf T}=\ <$$T, s, g , \tau $$>$, where $T,g,\tau$ are as above and $s$ is a special configuration called the {\it seed structure}. Let $C$ and $D$ be two configurations. Suppose that there exist some $t \in T$ and some $(x,y)\in{\bf Z}^2$ such that $D = C$ except that at $(x,y)$, $C(x,y)=null$ and $D(x,y)=t$. Let $f_{N,C,t}(x,y)=g(\sigma_N (t),\sigma_S (C(x,y+1))$. Informally, $f_{N,C,t}(x,y)$ is the strength of the bond on the north edge of $t$ in configuration $C$. We define $f_{S,C,t}(x,y), f_{E,C,t}(x,y)$ and $f_{W,C,t}(x,y)$ similarly. Then tile $t$ is {\it attachable} to $C$ at position $(x,y)$ iff $f_{C,t}(x,y) \equiv f_{N,C,t}(x,y)+f_{S,C,t}(x,y)+f_{E,C,t}(x,y)+f_{W,C,t}(x,y) \geq \tau$. We write $C\rightarrow_{\bf T} D$ to denote the transition from $C$ to $D$ by attaching a tile to $C$ at position $(x,y)$. Informally, $C\rightarrow_{\bf T} D$ iff $D$ can be obtained from $C$ by adding a tile $t$ such that the total strength of interaction between $t$ and $C$ is at least $\tau$. A {\it terminal assembly} is a configuration $A$ such that there is no configuration $B$ for which $A\rightarrow_{\bf T} B$.

When a tile $t$ attaches to configuration $C$ at position $(x,y)$, the edges $U$ of $t$ with $f_{U,C,t}(x,y)>0$ are called the {\it input} edges; all other edges are called the {\it output} edges. A tile system is {\it rectilinear} if there is a unique terminal assembly that can be reached starting from the seed structure, each tile $t$ has the same input and output edges every time it attaches, and all tiles have the same input and output edges.

\paragraph{The Kinetic Tile Assembly Model}
\label{sec:def_thermo}

According to the combinatorial tile assembly model, a tile $t$ attaches at a position $(x,y)$ in a configuration $C$ iff the total strength $f_{C,t}(x,y)$ of matched glues is at least $\tau$, and any tiles that attached never fall off. In practice, tiles may attach with a weaker binding strength, and tiles that already attached may fall off. These events can cause the tile system to behave differently from the combinatorial tile assembly model. We treat these deviations as errors and try to minimize the probability of having these events. In this paper, we use the kinetic tile assembly model proposed by Winfree~\cite{w98:phd} to model the {\it forward} and {\it reverse} rates, which are the rates at which a tile attaches to and falls off from a specific position, respectively.
This model computes these rates as functions of thermodynamic parameters as follows: 

\begin{enumerate}
\item The concentrations of the tiles are held constant throughout the self-assembly process.
\item The only two reactions allowed are single tiles attaching to and dissociating from a configuration.
\item The forward rate for tile $T_i$ is $k_fc_i$, where $k_f$ is a constant and $c_i$ is the concentration of tile $T_i$. This notation is used throughout this paper.
\item The reverse rate for a tile $t$ attached to configuration $C$ at position $(x,y)$ to fall off is $k_re^{-bG_{se}}$, where $k_r$ and $G_{se}$ are constants and $b=f_{C,t}(x,y)$ is the total strength of matched glues between $t$ and $C$.
\end{enumerate}

Here, the parameters $k_f$ and $k_r$ give the time scale of the self-assembly. The value of $G_{se}$ is determined by the binding strength of the sticky ends of DNA tiles. We use $c_{\max}$ and $c_{\min}$ to denote the maximum and minimum concentrations allowed in the tile system. If one wants the tile system to assemble according to the combinatorial tile assembly model most of the time, then the following two conditions need to hold. First, if the total binding strength between a tile and the original configuration it just attached to is less than $\tau$, then the tile must fall off quickly, i.e.,
\[
k_r e^{-(\tau -1)G_{se}}\  \gg \ k_fc_{\max}.
\]
Second, if a tile $t$ is attachable to a position in $C$, then the forward rate at which it attaches should be greater than the reverse rate at which it falls off, i.e.,
\[
k_fc_{\min}\ >\ k_r e^{-\tau G_{se}}.
\]
In practice, since each tile may have a slightly different value for the parameter $k_f$ and the strength of each glue may vary, one often needs to set $k_f$ and $k_r$ (by changing an experiment's temperature) such that 
\[
k_fc_{\min}\ \gg \ k_r e^{-\tau G_{se}}.
\]

For the remainder of the paper, we assume 
\[
k_r e^{-(\tau -1)G_{se}}\  \gg \ k_fc_{\max}\ >\ k_fc_{\min}\ \gg \ k_r e^{-\tau G_{se}}.
\]
We also assume that our seed structure is made by some other processes (e.g., DNA origami~\cite{r06:origami}) and its tiles never fall off.

\section{Minimizing the Error Rate}
\label{sec:error}

In this section, we consider the problem of changing the concentrations of tiles to minimize the failure probability (i.e., error rate) for a rectilinear tile system. There are three types of errors in tile self-assembly. A {\it growth error} refers to an incorrect tile attaching at a position instead of the correct tile~\cite{wb02:proof}. A {\it facet error} refers to an incorrect tile attaching at a position where no tile is supposed to attach~\cite{cg04:snaked}. A {\it nucleation error} refers to single tiles attaching to each other to form a lattice without the seed structure~\cite{sw04:zigzagtheory}. In this section, we only consider minimizing growth errors. It has been shown~\cite{wb02:proof} that if one slows down the growth of the boundary tiles of the target shape, the growth front will be roughly triangular. In this situation, growth errors are the dominant type of errors among the above three types.

We want to compute the probability that a tile $T_j$ causes a growth error by attaching at a position $(x,y)$ where only $T_i$ can attach with total binding strength at least $\tau$. First, $T_i$ can attach at that position at rate $k_fc_i$. Once $T_i$ is attached, the probably of it falling off is negligible since $k_fc_{\min} >> k_r e^{-\tau G_{se}}$. Second, $T_j$ can attach at that position at rate $k_fc_j$. When $T_j$ attaches, it can fall off at rate $k_r e^{-(\tau-m)G_{se}}$, where $m$ is the total strength of the mismatched glues between $T_i$ and $T_j$ on their input edges. $T_j$ can also get locked in place and cause an error due to the attachment of one or more adjacent tiles. The rate $r$ at which $T_j$ gets locked in place may vary with the features in the partially assembled shape near position $(x,y)$ such as long facets. In this paper, we assume that $r$ is the same for all positions $(x,y)$ and tiles $t$. The allowable reactions related to $T_i$ and $T_j$ are summarized in Figure~\ref{fig:markov_chain}. From the above description of reaction rates, we know that at a given position $(x,y)$ where tile $T_i$ is supposed to attach, the probability of having a growth error caused by $T_j$ is $\frac{c_j}{c_i}\epsilon_{ij}$. Here, $\epsilon_{ij} = \frac{r}{r+k_r e^{-(\tau-m)G_{se}}}$, which is  roughly at the order of $e^{-mG_{se}}$ since $r \leq 2k_fc_{\max}$. Therefore, at position $(x,y)$, the total probability of a growth error is 
\[
\frac{\sum_{j\neq i}\epsilon_{ij}c_j}{c_i}.
\]

\begin{figure*}[htbp]
    \centering
    {\includegraphics[scale=.30]{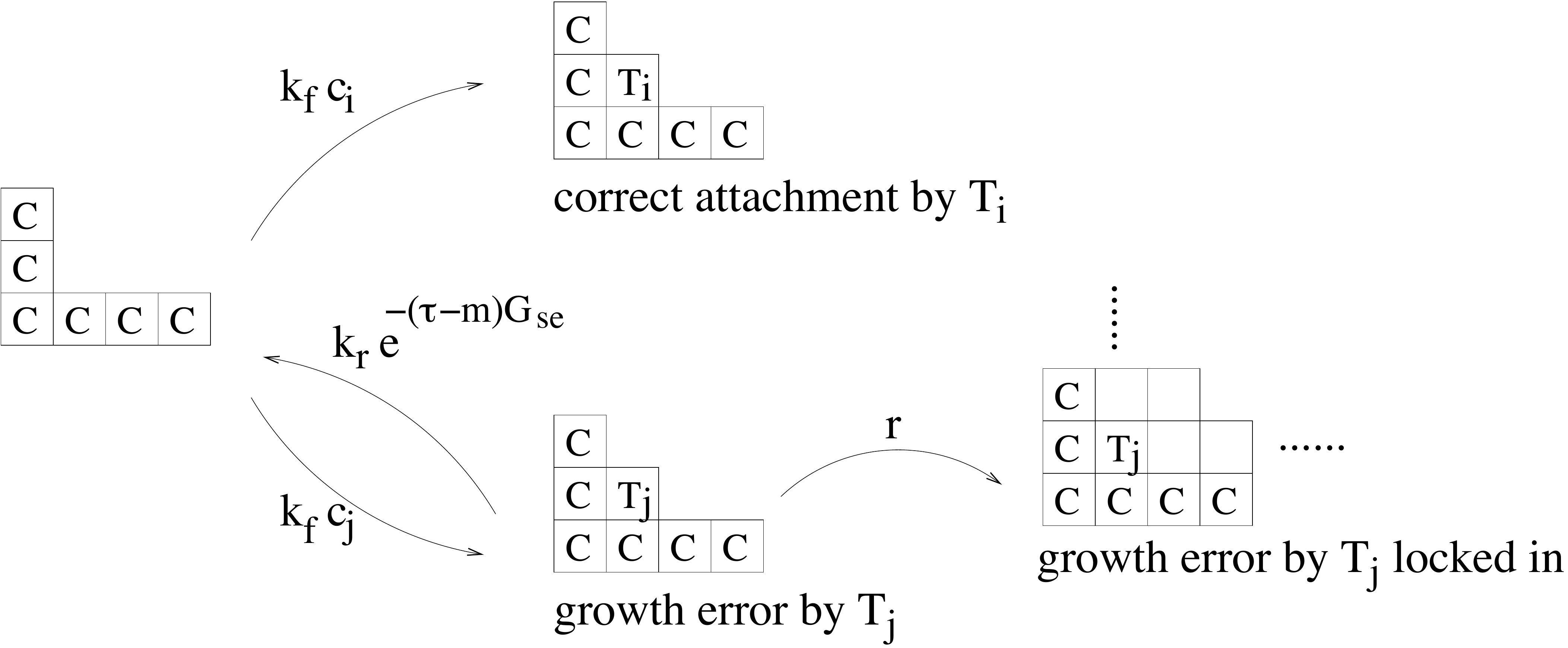}}
  \caption{\label{fig:markov_chain}
    \small{A Markov chain describing attachments of $T_i$ and $T_j$, where $C$ indicates a correct tile.}} 
\end{figure*}



For a self-assembly process, the error rates at different positions depend on each other. However, if one wants to have a high probability of success, one almost always needs to set the experimental condition such that the error rate at each position is much smaller than $1/n$. In this case, minimizing the sum of error rates over all positions is a good approximation of minimizing the actual overall error rate of the tile assembly system. Thus, in the remainder of this section, we will minimize 
\begin{equation}
\label{eqn:total_error}
\sum_i N_i \big(\frac{\sum_{j\neq i}\epsilon_{ij}c_j}{c_i} \big),
\end{equation}
where $N_i$ is the number of positions outside the seed structure to which $T_i$ is supposed to attach. 








\begin{theorem}
\label{thm:error}
For a rectilinear tile system with a unique terminal assembly, the error rate (i.e., probability of failure) is minimized when the concentration of each tile $T_i$ is proportional to $\sqrt{N_i}$, where $N_i$ is the number of times tile $T_i$ appears outside the seed structure in the correct terminal assembly of the tile system.
\end{theorem}

\begin{proof}
From Equation~\ref{eqn:total_error}, we can scale the tile concentrations $c_i$ without loss of generality such that $\sum_i c_i = 1$, and we need to solve the following minimization problem:

\[
\mbox{Minimize} \sum_i N_i \big(\frac{\sum_{j\neq i}\epsilon_{ij}c_j}{c_i} \big),
\]
\[
\mbox{subject to} \sum_i c_i\ =\ 1.
\]

The Lagrange multiplier for this minimization problem is
\[
\Lambda\ =\ \sum_i N_i \big(\frac{\sum_{j\neq i}\epsilon_{ij}c_j}{c_i} \big)\ +\ \lambda (\sum_i c_i - 1).
\]

We need to solve 
\begin{equation}
\label{eqn:derivative}
\frac{\partial \Lambda}{\partial c_i}\ =\ -N_i \big(\frac{\sum_{j\neq i}\epsilon_{ij}c_j}{c_i^2} \big)\ +\ \sum_{j\neq i}N_j \frac{\epsilon_{ij}}{c_j}\ +\ \lambda\ =\ 0 \ \ \ \mbox{for all $i$}
\end{equation}

and 
\[
\sum_i c_i\ =\ 1.
\]

Simplifying Equation~\ref{eqn:derivative}, we obtain 
\[
\sum_{j\neq i} \big(-N_i\frac{\epsilon_{ij}c_j}{c_i^2}\ + N_j\frac{\epsilon_{ij}}{c_j} \big)\ +\ \lambda\ =\ 0 \ \ \ \mbox{for all $i$,}
\]

and consequently the error rate is minimized when
\[
c_i = \frac{\sqrt{N_i}}{\sum_i \sqrt{N_i}}.
\]
\end{proof}

Two points about the proof of Theorem~\ref{thm:error} are worth noticing. Firs, the error rate only depends on the ratio between the tile concentrations. Therefore, the error rate is minimized when the concentration of $T_i$ is proportional to $\sqrt{N_i}$ even when we vary each $c_i$ between $c_{\max}$ and $c_{\min}$. Second, the same proof can apply to all tile systems that satisfy $\epsilon_{ij} = \epsilon_{ji}$ for all $i$, $j$. Hence Theorem~\ref{thm:error} is valid for many other systems already implemented in lab such as zig-zag ribbons~\cite{sw05:zigzagexp} and counters seeded by origami~\cite{bsrw09:counter_origami}. Moreover, Theorem~\ref{thm:error} still applies if we add either a uniform proofreading scheme~\cite{wb02:proof} or a snaked proofreading scheme~\cite{cg04:snaked} to the tile system. However, for a proofreading system to work, we need $k_f \cdot c_i$ to be sufficiently close to $k_r e^{-\tau G_{se}}$ for all tiles $T_i$, which is false in the parameter range that we use in this paper.

\section{Simulation Results for Theorem~\ref{thm:error}}
\label{sec:simulation}

We used a software called xgrow~\cite{xgrow} to simulate four tile systems to determine their error rates under different tile concentrations. To obtain a good estimate of the error rate of a tile system, we would choose our parameters such that errors can be frequently observed. However, in most tile systems, if we use such parameters, we will reach some configuration very different from the terminal assembly predicted by the combinatorial tile assembly model. Since our prediction of the optimal tile concentrations depends on the terminal assembly, we made a design decision to perform simulations on tile systems for which each error only affects one position of the terminal assembly.

We simulated four tile systems named as $A_1$, $A_2$, $B_1$, and $B_2$. Each of the four systems operates at $\tau = 2$ and only has two tiles $X$ and $Y$ shown in Figure~\ref{fig:system_diagram}(a) beside the seed structure. The only difference between the four systems is their seed structures. The seed structures of tile systems $A_1$ and $A_2$ are shown in Figure~\ref{fig:system_diagram}(b). The lengths of their seed structures are adjusted such that $N_X : N_Y = 25:1\ \mbox{and}\ 64:1$ for $A_1$ and $A_2$, respectively, where $N_X$ and $N_Y$ are the numbers of positions $X$ and $Y$ appear in the terminal assembly. The seed structures of tile systems $B_1$ and $B_2$ are shown in Figure~\ref{fig:system_diagram}(c). The lengths of their seed structures are also adjusted such that $N_X : N_Y = 25:1\ \mbox{and}\ 64:1$ for $B_1$ and $B_2$, respectively. These systems are rectilinear with all tiles having their input edges on the south and east edges. Since tiles $X$ and $Y$ have the same output edges, when an error happens, the error only affects the position where the erroneous tile is attached. The unique terminal assemblies and example configurations generated by simulations of tile systems $A_1$ and $B_1$ are shown in Figures~\ref{fig:systemA}~and~\ref{fig:systemB}, respectively. Theorem~\ref{thm:error} predicts that the rate of growth errors is minimized when $c_X : c_Y = 5:1$ for systems $A_1$ and $B_1$, and when $c_X : c_Y = 8:1$ for systems $A_2$ and $B_2$.

\begin{figure*}[htbp]
    \centering
    {\includegraphics[scale=.4]{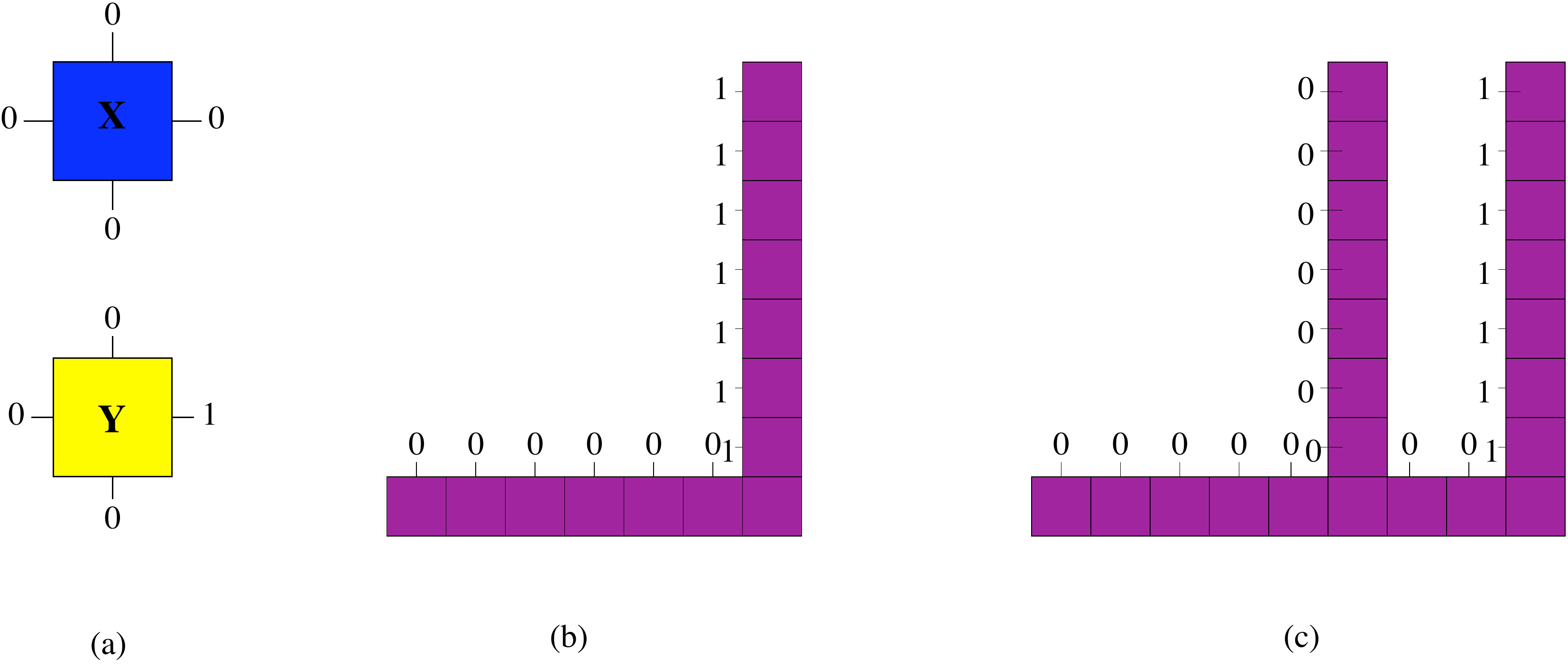}}
  \caption{\label{fig:system_diagram}
    \small{(a) Tiles $X$ and $Y$, where all glues have strength $1$. (b) An L-shaped seed structure for systems $A_1$ and $A_2$. (c) A seed structure with two vertical facets and one horizontal facet for systems $B_1$ and $B_2$.}} 
\end{figure*}

\begin{figure*}[htbp]
    \centering        
    {\includegraphics[scale=.30]{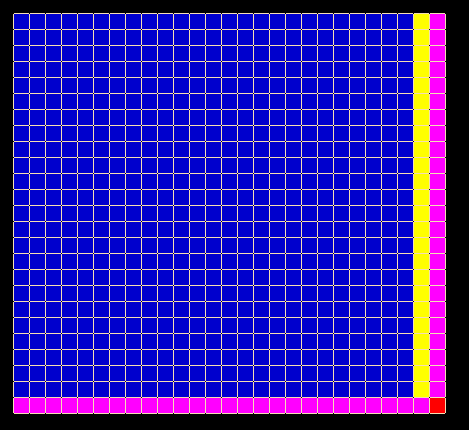}}        
    \hspace{2cm}
    {\includegraphics[scale=.30]{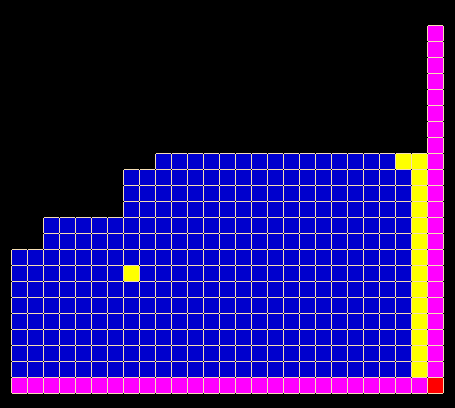}}
  \caption{\label{fig:systemA}
    \small{The left figure is the unique terminal assembly for system $A_1$. The right figure is an example configuration for system $A_1$ generated by an xgrow simulation.}} 
\end{figure*}

\begin{figure*}[htbp]
    \centering        
    {\includegraphics[scale=.30]{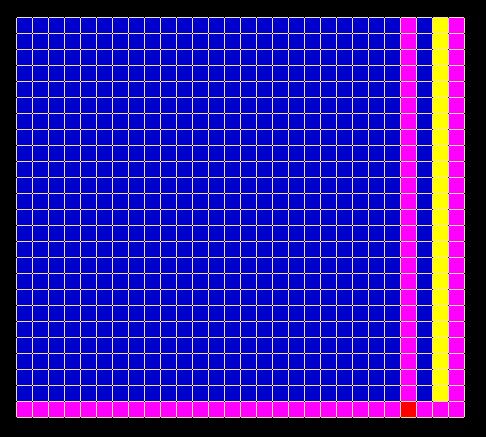}}
    \hspace{2cm}
    {\includegraphics[scale=.30]{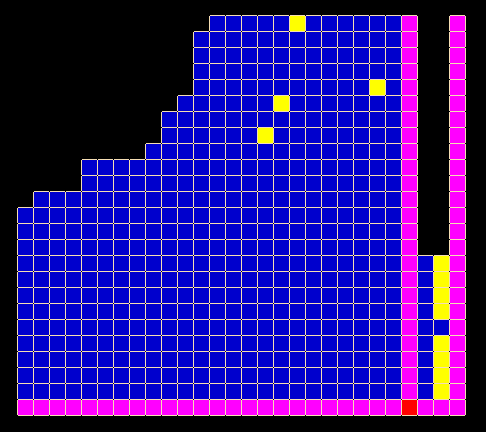}}
  \caption{\label{fig:systemB}
    \small{The left figure is the unique terminal assembly for system $B_1$. The right figure is an example configuration for system $B_1$ generated by an xgrow simulation.}} 
\end{figure*}

\begin{figure*}[htbp]
   \centering
    {\includegraphics[scale=.50]{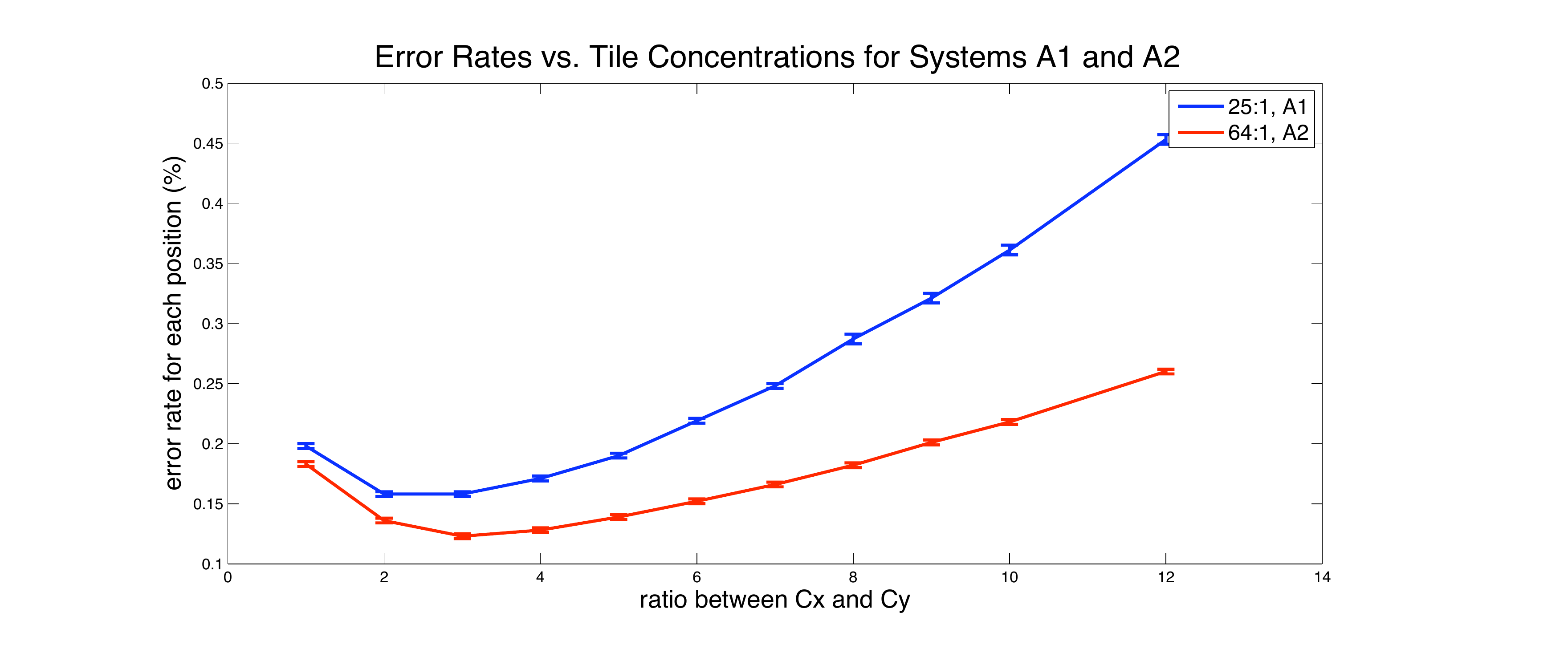}}
    {\includegraphics[scale=.50]{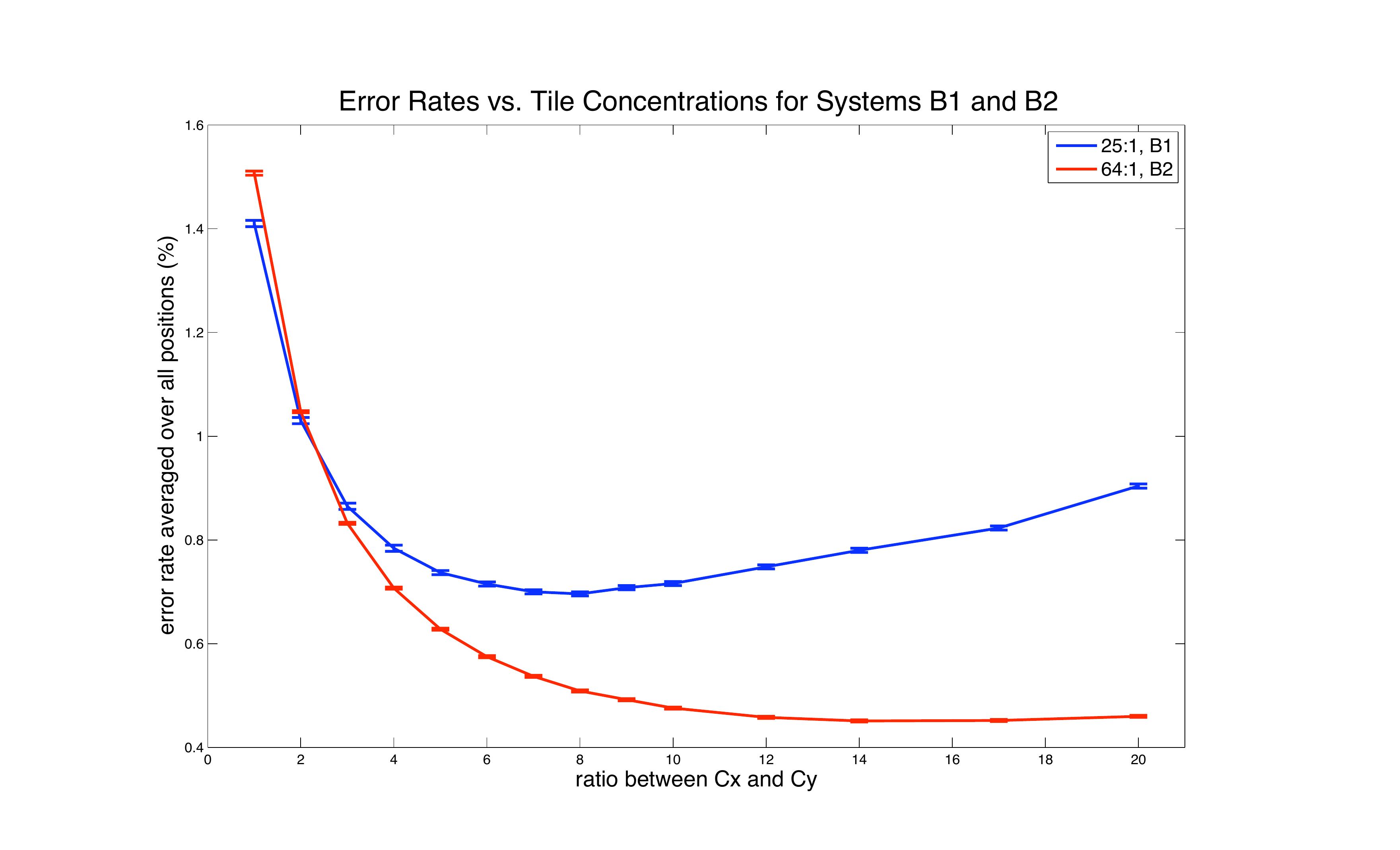}}
  \caption{\label{fig:error_rates}
    \small{Plots of error rates vs.~the ratios between tile concentrations. Each data point represents $m = 20,000$ simulations. The simulations use $G_{se}=9$ for systems $A_1$ and $A_2$, $G_{se}=11$ for systems $B_1$ and $B_2$, $c_X + c_Y = e^{-16}$, and $k_f = k_r$. Error bars show two standard deviations of the errors, computed using $\sigma = \sigma_{\mbox{simulation}}/\sqrt{m}$.}} 
\end{figure*}

The simulation results are shown in Figure~\ref{fig:error_rates}. In systems $A_1$ and $A_2$, the optimal tile concentration ratios are $2.5:1$ and $3:1$, respectively. In systems $B_1$ and $B_2$, the optimal tile concentration ratios are roughly $7.5:1$ and $15:1$, respectively. The major reason causing these deviations from the predictions made by Theorem~\ref{thm:error} is the facet errors. Since tiles $X$ and $Y$ both have glue $0$ on the south edge, having a long horizontal facet may introduce a large number of facet errors. For systems $A_1$ and $A_2$, notice that long horizontal facets are generated because tile $Y$ (colored yellow) has lower concentrations and grows slower than $X$. An example configuration for system $A_1$ that demonstrates these horizontal facets is shown in Figure~\ref{fig:systemA}. Such undesirable facets become longer and more when we increase the ratio between $c_X$ and $c_Y$. Therefore, the actual optimal tile concentration is biased towards having more $Y$ than predicted by Theorem~\ref{thm:error}. For systems $B_1$ and $B_2$, each terminal assembly is separated into a left portion and a right portion by the seed structure, as shown in Figure~\ref{fig:systemB}. Horizontal facets can only be generated in the left portion, where all tiles should be $X$. Hence, we can reduce facet errors by decreasing the concentration of tile $Y$ and thus the actual optimal tile concentration is biased towards having fewer $Y$ than predicted by Theorem~\ref{thm:error}.

\section{Minimizing the Expected Assembly Time}
\label{sec:time}

This section assumes that only the correct tiles can attach and any tile that already attached never falls off. We minimize the expected assembly time by varying the tile concentrations.

\begin{theorem}
\label{thm:time}
Consider any tile system with the four properties that 
%
%
%
%
%
\begin{enumerate}
\item tiles can only attach one by one,
\label{property1}
\item only the correct tiles can attach,
\label{property2}
\item any tile that already attached never falls off, and 
\label{property3}
\item there is a unique terminal assembly.
\label{property4}
\end{enumerate}
Assume that the total tile concentration is $\sum_i c_i = 1$. Setting $c_i = \frac{\sqrt{N_i}}{\sum_i \sqrt{N_i}}$ minimizes the expected assembly time of the tile system.
\end{theorem}
\begin{proof}
From the kinetic tile assembly model~\cite{w98:phd}, the expected time for a tile $T_i$ to attach to a location $(x,y)$ is $\frac{1}{k_fc_i}$. By the four properties of the given tile system in the theorem, the expected assembly time is $\frac{1}{k_f}\sum_i \frac{N_i}{c_i}$, where $N_i$ is the number of times $T_i$ appears outside the seed structure in the terminal assembly. We want to minimize the sum $\sum_i \frac{N_i}{c_i}$ subject to $\sum_i c_i =1$. 
The Lagrange multiplier for this minimization problem is
\[
\Lambda\ =\ \sum_i \frac{N_i}{c_i}\ +\ \lambda (\sum_i c_i - 1).
\]
Solving the equations
\[
\frac{\partial \Lambda}{\partial c_i}\ =\ -\frac{N_i}{c_i^2}\ +\ \lambda\ =\ 0  \ \ \ \mbox{for all $i$}
\]
and
\[
\sum_i c_i\ =\ 1,
\]
we obtain
\[
c_i = \frac{\sqrt{N_i}}{\sum_i \sqrt{N_i}}.
\]
\end{proof}

In settings that are more general than Theorem~\ref{thm:time} assumes, the optimal tile concentrations to minimize the expected assembly time may significantly deviate from the $c_i$'s determined in Theorem~\ref{thm:error}. Lemma~\ref{lem:example} describes an example of such deviation.

\begin{lemma}
\label{lem:example}
Consider a tile system for which properties~\ref{property2}~through~\ref{property4} in Theorem~\ref{thm:time} hold but property~\ref{property1} is replaced by allowing tiles to attach in parallel whenever they can. Also, as in Theorem~\ref{thm:time}, assume that the total tile concentration is $\sum_i c_i = 1$. Further assume that in the final assembly, there are $N$ tiles not in the seed structure, and all $N$ tiles can attach to the seed structure directly and in parallel. Then the optimal tile concentrations that minimize the expected assembly time must satisfy $c_i \geq \frac{1}{k\ln N}$, where $k$ is the number of tile types.
\end{lemma}
\begin{proof}
A key point of this proof is that the attachment of tiles to all $N$ positions outside the seed structure in the terminal assembly can happen in parallel. If all tiles have the same concentration $1/k$, then the expected assembly time is $k\ln N$. If any tile has concentration lower than $\frac{1}{k\ln N}$, then the attachment of that single tile non-optimally takes expected time more than $k\ln N$. Thus, when the expected assembly time is minimized, all tile concentrations must be at least $\frac{1}{k\ln N}$.
\end{proof}

For general tile assembly systems, we do not know how to analytically find the optimal tile concentrations to minimize the expected assembly time. However, we show in Theorem~\ref{thm:simulation} below that given a set of tile concentrations, only a polynomial number of simulations is required in order to approximate the expected assembly time with high accuracy and probability.

\begin{lemma}
\label{lem:simulation}
Consider any tile system for which properties~\ref{property2}~through~\ref{property4} in theorem~\ref{thm:time} hold but there is no assumption on whether tiles can attach only one by one or in parallel. If the assembly process of the tile system takes expected time $S$, then for any $\epsilon > 0$, the average of the assembly times over $48S^2\frac{1}{\epsilon}$ simulations of the assembly process will be between $S-\epsilon$ and $S+\epsilon$ with probability at least $3/4$.
\end{lemma}
\begin{proof}
We only need to bound the variance of the assembly time of a simulation to show that the average of the assembly times of simulations converges fast. Suppose the expected time for the final shape to assemble is $S$. From Markov's inequality, the probability of having the time of a simulation greater than $2S$ is at most $0.5$. Also, the configuration at time $2S$ contains the seed structure because tiles in the seed structure do not fall off. Therefore if the assembly process is not finished at time $2S$, then starting from time $2S$, the expected time for the process to finish is less than $S$. Thus, the probability that a simulation takes more than $2kS$ time is at most $2^{-k}$, and the variance of the assembly time of a simulation is at most $\sum_i \frac{1}{2^i}(2iS)^2\ =\ 24S^2$. Thus, for any given $\epsilon$, if we run $48S^2\frac{1}{\epsilon}$ simulations, then the probability that the average of the assembly times of the simulations lands between $S-\epsilon$ and $S+\epsilon$ is at least $3/4$.
\end{proof}

\begin{theorem}
\label{thm:simulation}
Consider any tile system for which properties~\ref{property2}~through~\ref{property4} in Theorem~\ref{thm:time} hold but there is no assumption on whether tiles can attach only one by one or in parallel. Let $N$ be the number of positions outside the seed structure in the terminal assembly. If the assembly process of the tile system takes expected time $S$, then for any $\epsilon > 0$, the average of the assembly times over $O(n^4\frac{1}{\epsilon\ c_{\mbox{{\tiny min}}}^2})$ simulations of the assembly process will be between $S-\epsilon$ and $S+\epsilon$ with probability at least $3/4$.
\end{theorem}
\begin{proof}
Goel et al.~\cite{cgm04:optcounter} showed that if all tiles have concentrations $c$, then the expected assembly time is $\Theta(d/c)$, where $d$ is the length of the longest path in the assembly order of the terminal assembly. Since reducing tile concentrations only slows down the assembly process, the expected assembly time $S$ is $O(d/c_{\min}) = O(n/c_{\min})$. The theorem follows from this upper bound on $S$ and Lemma~\ref{lem:simulation}.
\end{proof}

\section{Further Research}
\label{sec:conclusion}

In Section~\ref{sec:error}, we gave closed-form formulas to minimize the growth errors by varying the concentration of each tile. In Section~\ref{sec:simulation}, we found in simulations that facet errors are also an important factor that needs to be considered in lab implementations. At the theoretical level, it is open to find closed-form formulas or efficient algorithms to minimize the facet errors by varying the tile concentrations.

In Section~\ref{sec:time}, we gave closed-form formulas to minimize the expected assembly time for a feasible class of tile systems by varying the tile concentrations. For general tile systems, the best known algorithm  can compute an $O(\log n)$-approximation of the minimum expected assembly time~\cite{acghkmr02:opt}. It is of interest to determine whether one can compute the precise minimum expected assembly time or an estimate with a better approximation factor than $O(\log n)$. Given tile concentrations, we showed that simulations can accurately predict the expected assembly time with high probability. The time it takes to run the required simulations is polynomial in the size of the terminal assembly, not the tile system itself. It would be useful if one can approximate the expected assembly time just by analyzing the tile system and some succinct features of the terminal assembly (e.g., the number of times each tile appears outside the seed structure in the terminal assembly) in time polynomial in the size of the tile system.

\bibliographystyle{plain}
\bibliography{case}

\end{document}